\documentclass[letterpaper, 10 pt, conference]{ieeeconf}

\IEEEoverridecommandlockouts
\usepackage{geometry}
\newgeometry{top=0.75in, bottom=0.75in, right=0.75in, left=0.75in}

\usepackage[utf8]{inputenc}
\usepackage{amsmath, amssymb, amsfonts}
\usepackage{graphicx}
\usepackage{booktabs}
\usepackage{algorithm}
\usepackage{algorithmic}
\usepackage{float}
\usepackage{cite}
\usepackage{xcolor}
\usepackage[font=normalsize]{caption}
\usepackage[colorlinks=true, linkcolor=blue, citecolor=blue, urlcolor=blue]{hyperref}

\usepackage{amsthm}

\theoremstyle{plain}
\newtheorem{proposition}{Proposition}

\theoremstyle{definition}
\newtheorem{remark}{Remark}

\title{\LARGE \bf {From Points to Sets: Set-Based Safety Verification in the Latent Space}}

\author{Wenyuan Wu$^{1}$, Peng Xie$^{1}$, Zhen Zhang$^{1}$, Yanliang Huang$^{1}$, Karl H. Johansson$^{2}$, and Amr Alanwar$^{1}$%
\thanks{$^{1}$ School of Computation, Information and Technology, Technical University of Munich, Germany.
        (Email: {\tt\small \{wenyuan.wu, p.xie, zhenzhang.zhang, yanliang.huang, alanwar\}@tum.de})}%
\thanks{$^{2}$ Division of Decision and Control Systems, KTH Royal Institute of Technology, Stockholm, Sweden.
        (Email: {\tt\small kallej@kth.se})}%
}

\begin{document}

\maketitle
\thispagestyle{empty}
\pagestyle{empty}

\begin{abstract}
We extend latent representation methods for safety control design to set-valued states. Recent work has shown that barrier functions designed in a learned latent space can transfer safety guarantees back to the original system, but these methods evaluate certificates at single state points, ignoring state uncertainty. A fixed safety margin can partially address this but cannot adapt to the anisotropic and time-varying nature of the uncertainty gap across different safety constraints. We instead represent the system state as a zonotope, propagate it through the encoder to obtain a latent zonotope, and evaluate certificates over the worst case of the entire set. On a 16-dimensional quadrotor suspended-load gate passage task, set-valued evaluation achieves 5/5 collision-free passages, compared to 1/5 for point-based evaluation and 2/5 for a fixed-margin baseline. Set evaluation reports safety in 44.4\% of per-head evaluations versus 48.5\% for point-based, and this greater conservatism detects 4.1\% blind spots where point evaluation falsely certifies safety, enabling earlier corrective control. The safety gap between point and set evaluation varies up to $12\times$ across certificate heads, explaining why no single fixed margin suffices and confirming the need for per-head, per-timestep adaptation, which set evaluation provides by construction.
\end{abstract}

\section{Introduction}

Safety-critical control of autonomous systems requires formal guarantees that the system will not violate prescribed safety constraints. Control barrier functions (CBFs) provide such guarantees by enforcing forward invariance of a safe set \cite{c1}. Despite their strong theoretical properties, constructing valid CBFs for high-dimensional nonlinear systems remains challenging due to computational intractability and scalability limitations \cite{c2}.

\begin{figure*}[!t]
\centering
\includegraphics[width=\textwidth]{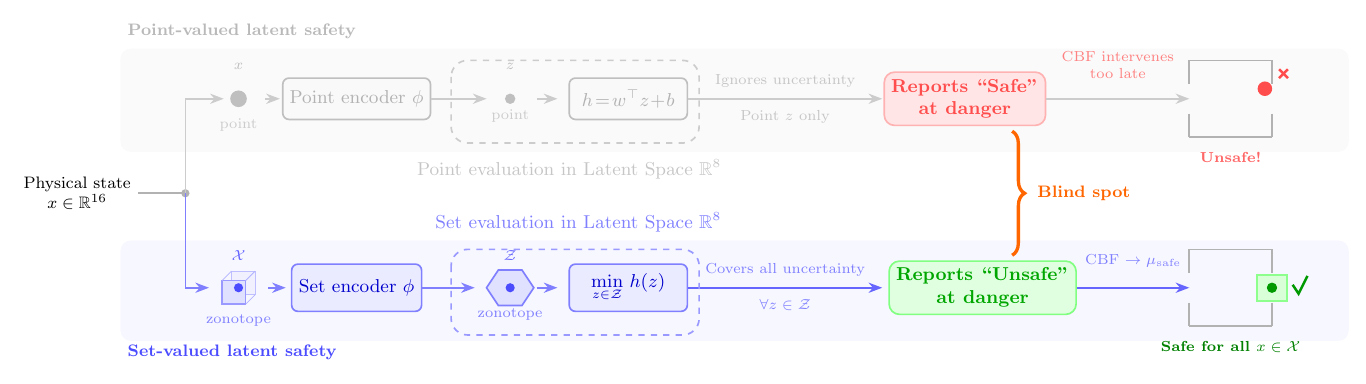}
\caption{Overview of point-valued vs.\ set-valued latent safety evaluation. Both tracks encode the same physical state $x$. Top (gray): point evaluation ignores state uncertainty, reports ``Safe'' at a dangerous moment (blind spot), and causes the CBF to intervene too late to prevent collision. Bottom (blue): set evaluation covers all uncertainty via $\min_{z \in \mathcal{Z}} h(z)$, correctly reports ``Unsafe'' at the dangerous moment and triggers $\mu_{\text{safe}}$ early enough for safe passage.}
\label{fig:pipeline}
\end{figure*}

Recent work by Lutkus et al.~\cite{c3} addresses this challenge by leveraging learned latent representations. In this framework, an encoder maps the high-dimensional system state to a low-dimensional latent space where barrier functions can be more efficiently constructed. Under approximate conjugacy conditions, safety guarantees established in the latent space can be transferred back to the original system. This approach significantly improves scalability and has been demonstrated on systems with up to six dimensions.

However, a key limitation of the work in~\cite{c3} and related latent-space approaches~\cite{c4,c5} is their reliance on pointwise evaluation of safety certificates. In practice, the system state is never known exactly: sensor noise, state estimation error, and model uncertainty imply that the true state lies within a bounded uncertainty set rather than at a single point. Consequently, evaluating a safety certificate only at the nominal estimate can lead to false assurances of safety, as nearby states within the uncertainty set may violate safety constraints.

This limitation is particularly critical in safety-sensitive applications, where robustness to uncertainty is essential. A certificate that is valid only at a single point does not, in general, guarantee safety for the true system state. Instead, safety must be verified over the entire set of states consistent with the available information.

CBFs were formalized by Ames et al.~\cite{c1,c8} and have since become a cornerstone of safety-critical control. A wide range of approaches have been developed to address the challenge of constructing CBFs in complex systems, including sum-of-squares programming \cite{c2}, Hamilton--Jacobi reachability analysis \cite{c9}, and learning-based methods \cite{c10}. More recently, latent-space methods have emerged as a promising direction: Castaneda et al.~\cite{c4} learn in-distribution barrier functions, while Kumar et al.~\cite{c5} design CBFs directly in learned latent representations. Lutkus et al.~\cite{c3} further provide formal guarantees for transferring safety certificates between latent and original spaces. Despite these advances, all existing latent approaches evaluate certificates at single state points.

In parallel, zonotope-based set representations have been widely used for scalable reachability analysis~\cite{zhang2025data,huang2026conformalized} and neural network verification \cite{c7,c11}. These methods enable sound over-approximation of uncertainty propagation through nonlinear transformations, including learned models. However, they have not been integrated with latent-space safety certification.

In this work, we bridge this gap by introducing a set-valued extension of latent safety certification. We represent the system state as a zonotope capturing bounded uncertainty, propagate this set through a learned encoder to obtain a latent set, and evaluate the safety certificate over the worst-case element of this set. This set-based evaluation yields sufficient conditions that guarantee safety for all states within the uncertainty set, rather than only the nominal estimate. As a result, our approach provides robust safety guarantees that explicitly account for state uncertainty while retaining the scalability benefits of latent representations.

We extend the latent guarantee framework of~\cite{c3} to set-valued state representations (Fig.~\ref{fig:pipeline}). More specifically, we make the following contributions.

\begin{enumerate}
    \item We represent states as zonotopes and propagate them through a learned encoder using sound set-based forward passes, yielding latent zonotope sets.
    \item We evaluate barrier certificates over the worst case of the entire latent set, providing sufficient safety conditions for all states within the uncertainty zonotope.
    \item We demonstrate on a 16-dimensional quadrotor suspended-load system that set-valued evaluation achieves 5/5 collision-free passages, compared to 1/5 for point-based and 2/5 for a fixed-margin baseline. Set evaluation detects 4.1\% blind spots where point evaluation falsely certifies safety, and the safety gap varies up to $12\times$ across certificate heads, confirming the need for adaptive per-head, per-timestep margins, which set evaluation provides by construction.
\end{enumerate}

\section{Preliminaries and Problem Formulation}

\subsection{Zonotopes}
%\subsection{Notation}

We denote sets with calligraphic letters (e.g., $\mathcal{X}$, $\mathcal{Z}$, $\mathcal{C}$). For a vector $x \in \mathbb{R}^n$, $\|x\|$ denotes the Euclidean norm.

A zonotope $\mathcal{Z} \subset \mathbb{R}^n$ is a convex set defined by a center $c \in \mathbb{R}^n$ and generator matrix $G \in \mathbb{R}^{n \times q}$:
\begin{equation}
\mathcal{Z} = \langle c, G \rangle = \left\{ c + \sum_{i=1}^{q} G_{(\cdot,i)} \beta_i \;\middle|\; \beta \in [-1,1]^q \right\}
\end{equation}
For a linear function $h(z) = w^\top z + b$, the minimum over a zonotope is computed exactly~\cite{c6}:
\begin{equation}
\min_{z \in \mathcal{Z}} h(z) = w^\top c + b - \sum_{j=1}^{q} |w^\top G_{(\cdot,j)}|
\label{eq:zono_min}
\end{equation}
An affine map $f(z) = Wz + b$ applied to a zonotope yields a zonotope: $f(\mathcal{Z}) = \langle Wc + b, \; WG \rangle$.

\subsection{Zonotope Propagation Through Neural Networks}
We propagate zonotopes through neural network layers. Linear layers apply affine maps. For nonlinear layers (e.g., tanh), we compute a sound over-approximation using a chord-slope linear bound~\cite{c7} with analytically computed error intervals:

For $\varphi = \tanh$ on interval $[l, u]$, the slope $m = (\tanh(u) - \tanh(l))/(u-l)$ yields approximation errors $[\underline{e}, \bar{e}]$ (computed via~\cite[Proposition 10]{c7}). The output zonotope is:
\begin{equation}
\varphi(\mathcal{Z}) \subseteq \langle m \odot c + \tfrac{1}{2}(\bar{e}+\underline{e}), \; [m \odot G, \; \text{diag}(\tfrac{1}{2}(\bar{e}-\underline{e}))] \rangle,
\end{equation}
where $\odot$ denotes element-wise multiplication and $[\cdot, \cdot]$ denotes horizontal concatenation. This enclosure is sound: $\varphi(\mathcal{Z}) \subseteq$ output zonotope.

\subsection{Control Barrier Functions}
Consider a control-affine system $\dot{x} = f(x) + g(x)u$. A function $h: \mathbb{R}^n \to \mathbb{R}$ is a CBF on a set $\mathcal{C} = \{x \mid h(x) \geq 0\}$ if there exists $\alpha > 0$ such that for all $x \in \mathcal{C}$~\cite{c1}:
\begin{equation}
\sup_{u \in \mathcal{U}} \left[ L_f h(x) + L_g h(x) u + \alpha h(x) \right] \geq 0
\end{equation}
where $L_f h = \nabla h \cdot f$ and $L_g h = \nabla h \cdot g$ are the Lie derivatives. This condition ensures forward invariance of $\mathcal{C}$.

\subsection{Latent Representations for Control (Review of~\cite{c3})}
An encoder $E: \mathbb{R}^{n_x} \to \mathbb{R}^{n_z}$ maps the state to a latent space where barrier functions $h: \mathbb{R}^{n_z} \to \mathbb{R}$ are defined. The induced barrier in original space is $\bar{h}(x) = h(E(x))$. Under $\varepsilon$-forward-conjugacy conditions~\cite[Definition 1]{c3}, if $h$ is an $L$-Lipschitz $\gamma$-barrier function in latent space, then the set $\mathcal{C}_x^\omega = \{x \mid \bar{h}(x) \geq -L\varepsilon/\gamma\}$ is forward invariant~\cite[Theorem 3]{c3}.

\subsection{Problem Formulation}
Consider a control-affine system $\dot{x} = f(x) + g(x)u$ with state $x \in \mathbb{R}^{n_x}$ and input $u \in \mathcal{U}$. Due to sensor noise and estimation errors, the true state is not known exactly; instead, it is known to lie within a bounded uncertainty set $\mathcal{X} = \langle x_c, G_x \rangle \subset \mathbb{R}^{n_x}$, represented as a zonotope centered at the nominal estimate~$x_c$.

Given a learned encoder $E: \mathbb{R}^{n_x} \to \mathbb{R}^{n_z}$ and linear barrier functions $h_i(z) = w_i^\top z + b_i$ defined in the latent space, the goal is to find a control input $u$ such that:
\begin{equation}
h_i(E(x)) \geq 0 \quad \forall x \in \mathcal{X}, \quad \forall i = 1, \ldots, N_h
\end{equation}
That is, safety must be guaranteed not only at the nominal state estimate $x_c$, but for \textit{all} states consistent with the current uncertainty. Existing approaches~\cite{c3,c4,c5} verify only $h_i(E(x_c)) \geq 0$, which provides no safety statement for $x \neq x_c$.

\section{Set-Valued Latent Safety Guarantees}

We propose a set-valued latent CBF controller that replaces point evaluation with worst-case evaluation over the state uncertainty set. Algorithm~\ref{alg:setcbf} summarizes the procedure; the remainder of this section details each step and establishes the safety guarantee.

\begin{algorithm}[H]
\caption{Set-Valued Latent CBF Controller}
\label{alg:setcbf}
\begin{algorithmic}[1]
\REQUIRE State $x$, nominal input $\mu_{\text{nom}}$, encoder $E$, certificates $\{(w_i, b_i)\}_{i=1}^3$, latent dynamics $A, B$, CBF parameter $\alpha$, error bounds $\{\varepsilon_{\text{dir},i}\}_{i=1}^3$
\ENSURE Safe input $\mu_{\text{safe}}$
\STATE Build state zonotope $\mathcal{X} = \langle x, \text{diag}(\varepsilon) \rangle$
\STATE Propagate through encoder: $\mathcal{Z} = E(\mathcal{X}) = \langle c_z, G_z \rangle$
\FOR{$i = 1, 2, 3$}
    \STATE $h_{\min,i} \leftarrow w_i^\top c_z + b_i - \sum_j |w_i^\top G_{z,j}|$
    \STATE $L_{f,\min,i} \leftarrow \min_{z \in \mathcal{Z}} w_i^\top A z$
\ENDFOR
\STATE Solve QP: $\mu_{\text{safe}} = \arg\min_\mu \|\mu - \mu_{\text{nom}}\|^2$
\STATE \quad s.t. $w_i^\top B \mu + L_{f,\min,i} + \alpha h_{\min,i} \geq \varepsilon_{\text{dir},i}$
\RETURN $\mu_{\text{safe}}$
\end{algorithmic}
\end{algorithm}

\subsection{Set-Valued State Representation}
We represent state uncertainty as a zonotope $\mathcal{X} = \langle x_c, G_x \rangle \subset \mathbb{R}^{n_x}$, where $x_c$ is the nominal (estimated) state and $G_x = \text{diag}(\varepsilon_1, \ldots, \varepsilon_{n_x})$ encodes per-dimension measurement uncertainty.

\subsection{Set-Valued Latent Encoding}
Given a learned encoder $E$, trained to minimize a composite loss of reconstruction, dynamics prediction, zonotope tightness, and safety-semantic teacher alignment (details in Section~IV-C), composed of alternating dense and tanh layers, we compute the latent set:

\begin{equation}
\mathcal{Z} = E(\mathcal{X}) = \langle c_z, G_z \rangle \subset \mathbb{R}^{n_z}
\end{equation}
via the sound zonotope propagation of Section II-B. By construction, $E(x) \in \mathcal{Z}$ for all $x \in \mathcal{X}$.

\subsection{Set-Valued Certificate Evaluation}
For a linear barrier $h(z) = w^\top z + b$, we evaluate:
\begin{equation}
h_{\min}(\mathcal{Z}) := \min_{z \in \mathcal{Z}} h(z) = w^\top c_z + b - \sum_{j} |w^\top G_{z,j}|
\end{equation}
This is the key difference from~\cite{c3}: where~\cite{c3} evaluates $h(E(x_c))$ at the center only, we evaluate $h_{\min}(E(\mathcal{X}))$ over the entire state uncertainty set. The safety condition becomes:
\begin{equation}
h_{\min}(\mathcal{Z}) \geq 0 \quad \Longrightarrow \quad h(E(x)) \geq 0 \;\; \forall x \in \mathcal{X}
\end{equation}

\subsection{Set-Valued CBF Condition}
The CBF constraint is similarly evaluated over the worst case. For linear $h(z) = w^\top z + b$ with latent dynamics $z^+ = Az + Bu$, both $L_f h(z) = w^\top A z$ and $h(z)$ are linear in $z$, so their combination $L_f h(z) + \alpha h(z) = w^\top(A + \alpha I)z + \alpha b$ is also linear. The exact worst-case over $\mathcal{Z}$ is therefore:
\begin{equation}
\min_{z \in \mathcal{Z}} \left[ w^\top(A + \alpha I) z \right] + \alpha b + w^\top B u \geq \varepsilon_{\text{dyn}}
\end{equation}
where the minimum is computed exactly via the zonotope formula of Section II-A. The term $\varepsilon_{\text{dyn}}$ accounts for latent dynamics prediction error, computed as a directed error bound (see Section IV-B). In Algorithm~\ref{alg:setcbf}, we evaluate $\min_{z \in \mathcal{Z}} w^{\top} A z$ and $\min_{z \in \mathcal{Z}} h(z)$ separately; since $\min(f+g) \geq \min(f) + \min(g)$, this yields a sound (potentially more conservative) bound.

\subsection{Safety Guarantee}
The following result is a direct consequence of the set-valued evaluation combined with~\cite[Theorem 3]{c3}.

\begin{proposition}[Set Safety Guarantee]
Let $h(z) = w^\top z + b$ be $L$-Lipschitz with $L = \|w\|$. Let $\mathcal{X} = \langle x_c, G_x \rangle$ be a state zonotope, $\mathcal{Z} = E(\mathcal{X})$ the latent set obtained via sound zonotope propagation~\cite{c7}, and $(f_z, \mathcal{D}_x)$ be $\varepsilon_{\text{conj}}$-forward-conjugate~\cite[Definition 1]{c3}. Suppose the CBF-QP finds $u$ satisfying:
\begin{equation}
L_g h \cdot u + \min_{z \in \mathcal{Z}}\left[L_f h(z) + \alpha h(z)\right] \geq \varepsilon_{\text{dyn}}
\end{equation}
Then for all $x \in \mathcal{X}$:
\begin{equation}
L_g \bar{h}(x) \cdot u + L_f \bar{h}(x) + \alpha \bar{h}(x) \geq \varepsilon_{\text{dyn}} - L\varepsilon_{\text{conj}}
\end{equation}
where $\bar{h} = h \circ E$, and $\varepsilon_{\text{conj}}$ is the forward conjugacy error.
\hfill $\lrcorner$
\end{proposition}

\textit{Proof.} We establish the result in three steps: set containment, worst-case bound, and conjugacy transfer.

\textit{Step 1 (Set containment).} The latent set $\mathcal{Z} = E(\mathcal{X})$ is a sound over-approximation of the true encoder image $\{E(x) \mid x \in \mathcal{X}\}$, obtained via the zonotope propagation through the neural network layers described in Section~II-B~\cite{c7}. In particular:

\begin{equation}
\forall x \in \mathcal{X}: \quad E(x) \in \mathcal{Z}
\end{equation}

\textit{Step 2 (Latent set inequality).} Since $L_f h(z) + \alpha h(z) = w^\top(A+\alpha I)z + \alpha b$ is linear in $z$, evaluating at any $E(x) \in \mathcal{Z}$ yields:
\begin{equation}
L_f h(E(x)) + \alpha h(E(x)) \geq \min_{z \in \mathcal{Z}}\left[L_f h(z) + \alpha h(z)\right]
\end{equation}
Combining with the CBF-QP constraint and using $L_g h \cdot u = w^\top B u$ (constant over $\mathcal{Z}$):
\begin{equation}
\forall x \in \mathcal{X}: \quad L_g \bar{h}(x) \cdot u + L_f \bar{h}(x) + \alpha \bar{h}(x) \geq \varepsilon_{\text{dyn}}
\end{equation}
where $\bar{h} = h \circ E$.

\textit{Step 3 (Conjugacy gap).} The latent dynamics $z^+ = Az + Bu$ are only an approximation of the true next latent state $E(f(x, u))$. Under $\varepsilon_{\text{conj}}$-forward-conjugacy~\cite[Definition 1]{c3}, this approximation error is bounded by $\varepsilon_{\text{conj}}$. Since $h$ is $L$-Lipschitz, the one-step safety transfer incurs a correction of at most $L\varepsilon_{\text{conj}}$, yielding the stated bound via~\cite[Theorem 3]{c3}. $\square$

\begin{remark}
The key advantage over point-based evaluation is that Proposition 1 provides a sufficient safety condition for the \textit{entire} uncertainty set $\mathcal{X}$, not just the center $x_c$. Point-based evaluation provides no safety statement for $x \neq x_c$. This safety guarantee is at the evaluation level: the set-valued certificate is sound over $\mathcal{X}$ by construction (exact zonotope arithmetic). The transfer of this guarantee to the physical system additionally requires a bounded conjugacy gap, discussed in the numerical instantiation below.
\hfill $\lrcorner$
\end{remark}

\textbf{Numerical instantiation.} On the trained model, the worst-case conjugacy error is $\varepsilon_{\text{conj}} = 2.46$ (maximum of $\|E(f(x,u)) - (AE(x)+Bu)\|$ over 547 near-gate validation steps), yielding per-head margins $\varepsilon_{\text{dir},i} - L_i \varepsilon_{\text{conj}}$ of $-2.33$, $-1.48$, and $-0.67$ for $h_z$, $h_y$, $h_E$ respectively. The negative margins indicate that the current latent dynamics model does not yet achieve formally certified safety transfer via Proposition~1. Nevertheless, the set-valued evaluation provides empirically verified safety improvements (Section~V-B). The training objective already includes a dynamics prediction loss $\mathcal{L}_{\text{dyn}} = \|Az + B\mu - z_{\text{next}}\|^2$ (Section~IV-C) that directly penalizes the conjugacy error. The residual gap arises primarily from extreme pendulum states near the training distribution boundary, where the linear latent dynamics model cannot fully capture the nonlinear system response. Closing this gap, for example via targeted data augmentation in high-swing regimes or certified Lipschitz encoder architectures, is an important direction for future work.

\section{Implementation}

\subsection{System Description}
We consider a 3D quadrotor ($m_q = 1.0$ kg, $n_x = 16$) carrying a suspended load ($m_L = 0.3$ kg) via a rigid rod ($L_{\text{rod}} = 0.8$ m), modeled as a 3D spherical pendulum with angles $(\alpha, \beta)$. The full state vector is:
\begin{equation}
x = [p_x, p_y, p_z, v_x, v_y, v_z, \phi, \theta, \psi, \omega_x, \omega_y, \omega_z, \alpha, \beta, \dot\alpha, \dot\beta]^\top
\end{equation}
The load position is coupled to the quadrotor through the pendulum kinematics:
\begin{equation}
p_{\text{load}} = p_{\text{quad}} + L_{\text{rod}} \begin{bmatrix} \sin\alpha \cos\beta \\ \sin\alpha \sin\beta \\ -\cos\alpha \end{bmatrix}
\end{equation}
and the pendulum dynamics are driven by the quadrotor acceleration through gravitational and inertial coupling, yielding a highly nonlinear 16-dimensional system.

The system has an inner-loop attitude controller (PD at 50 Hz), so the outer-loop control input is the acceleration command $\mu = [a_x, a_y, a_z] \in \mathbb{R}^3$. A rectangular gate ($1.2 \times 1.2$ m opening) is placed at $p_x = 10$ m. Safety requires that both the quadrotor body (collision radius 0.15 m) and the suspended load (collision radius 0.05 m) pass through the gate opening without collision.

Compared to the 6-dimensional inverted pendulum used in~\cite{c3}, this system presents three additional challenges: (i) the state dimension is $16$ (vs.\ $6$), requiring the encoder to compress $2.7\times$ more information; (ii) the pendulum-load coupling introduces nonlinear safety constraints that cannot be expressed as simple distance thresholds; and (iii) both the quadrotor and the suspended body must independently satisfy the gate clearance constraint, creating a multi-body safety requirement.

\subsection{Architecture}
\textbf{Encoder.} A 3-layer MLP (Dense-Tanh-Dense-Tanh-Dense) maps $\mathbb{R}^{16} \to \mathbb{R}^{8}$. During set evaluation, zonotopes are propagated through each layer using the constructions of Section II-B.

\textbf{Safety certificates.} Three linear barrier functions in latent space:
\begin{itemize}
    \item $h_z(z) = w_z^\top z + b_z$: z-direction (vertical) gate clearance
    \item $h_y(z) = w_y^\top z + b_y$: y-direction (lateral) gate clearance
    \item $h_E(z) = -w_E^\top z + (E_{\max} - b_E)$: pendulum swing energy bound
\end{itemize}
The weights $w_y, w_z, w_E$ are linear probes fit on the learned latent representation using near-gate training data.

\textbf{Directed dynamics error bound.} Rather than the conservative box bound $|w|^\top \cdot \varepsilon_{\text{dyn}}$, we compute a tighter directed bound by projecting the dynamics residual onto each certificate direction, reducing over-approximation by 1.6--3.7$\times$ (Table~\ref{tab:directed}):
\begin{equation}
\varepsilon_{\text{dir},i} = \text{quantile}_{99.5\%}\left( |w_i^\top (z_{\text{next}}^{\text{true}} - z_{\text{next}}^{\text{pred}})| \right)
\end{equation}

\begin{table}[!htbp]
\caption{Directed vs.\ box dynamics error bounds.}
\label{tab:directed}
\centering
\footnotesize
\begin{tabular}{lccc}
\toprule
\textbf{Head} & $\varepsilon_{\text{box}}$ & $\varepsilon_{\text{dir}}$ & \textbf{Reduction} \\
\midrule
$h_z$ (vertical) & 0.5811 & 0.3579 & 1.6$\times$ \\
$h_y$ (lateral)  & 0.5935 & 0.1604 & 3.7$\times$ \\
$h_E$ (energy)   & 0.2418 & 0.0860 & 2.8$\times$ \\
\bottomrule
\end{tabular}
\end{table}

\textbf{CBF-QP.} At each timestep:
\begin{equation}
\begin{aligned}
\min_{\mu} \;\; & \|\mu - \mu_{\text{nom}}\|^2 \\
\text{s.t.} \;\; & w_i^\top B \mu + \min_{z \in \mathcal{Z}}\!\left[w_i^\top\!(A + \alpha I) z\right] \\
& + \alpha b_i - \varepsilon_{\text{dir},i} \geq 0, \quad i = 1,2,3
\end{aligned}
\label{eq:cbf_qp}
\end{equation}

\begin{remark}[Control influence computation]
Proposition~1 uses the learned $B$ for the Lie derivative $L_g h_i = w_i^\top B$. Our implementation instead computes $L_g h_i$ via finite differences on the known plant dynamics, which yields the true control sensitivity of the composed certificate $\bar{h}_i = h_i \circ E$ up to discretization error $O(\delta^2)$. This does not invalidate the safety assessment: the set-valued terms ($h_{\min}$, $L_{f,\min}$) remain computed in the learned latent space via zonotope arithmetic, and the QP constraint structure is identical. The substitution replaces one source of approximation error (learned $B$) with a more accurate computation, and the dynamics error bound $\varepsilon_{\mathrm{dir}}$ already accounts for the residual between predicted and true latent dynamics.
\hfill $\lrcorner$
\end{remark}

\subsection{Training}
The encoder is trained on 10{,}000 state-action-next\_state tuples $(x_t, u_t, x_{t+1})$, collected via four complementary sampling strategies: random snapshots (25\%), trajectory rollouts (30\%), near-gate focused samples (25\%), and boundary examples near the safety threshold (20\%). Each sample is augmented with a 6-dimensional teacher signal $\phi_{\text{manual}}(x) \in \mathbb{R}^6$ encoding hand-crafted safety semantics: gate distance, vertical clearance, lateral clearance, pendulum swing energy, attitude margin, and a load-gate-clearance metric.

The total loss combines reconstruction fidelity, latent dynamics prediction, certificate head alignment, and set tightness regularization:
\begin{equation}
\mathcal{L}_{\text{rec}} = (1-\tau)\|D(E(\mathcal{X}))_c - x\|^2 + \tau \cdot \text{froRadius}(D(E(\mathcal{X})))
\end{equation}
\begin{equation}
\mathcal{L}_{\text{tight}} = \|G_z\|_F^2
\end{equation}
\begin{equation}
\mathcal{L}_{\text{dyn}} = \|Az + B\mu - z_{\text{next}}\|^2
\end{equation}
\begin{equation}
\mathcal{L}_{\text{teach}} = \|W_T z + b_T - \phi_{\text{manual}}\|^2
\end{equation}
The total training objective is:
\begin{equation}
\mathcal{L} = \lambda_{\text{rec}} \mathcal{L}_{\text{rec}} + \lambda_{\text{tight}} \mathcal{L}_{\text{tight}} + \lambda_{\text{dyn}} \mathcal{L}_{\text{dyn}} + \lambda_{\text{teach}} \mathcal{L}_{\text{teach}} + \mathcal{L}_{\text{reg}}
\label{eq:total_loss}
\end{equation}
where $\mathcal{L}_{\text{reg}}$ aggregates auxiliary regularization terms (right-inverse, isometry, contrastive separation, causal alignment, and control sensitivity losses) with small weights ($\lambda \leq 0.5$). Training uses Adam with learning rate $5 \times 10^{-4}$, batch size 256, for 120 epochs. Loss weights are $\lambda_{\text{rec}} = \lambda_{\text{dyn}} = 1.0$, $\lambda_{\text{tight}} = 0.5$, $\lambda_{\text{teach}} = 5.0$. The tightness loss $\mathcal{L}_{\text{tight}}$ is a key addition beyond~\cite{c3}: without it, the zonotope generators grow through the tanh layers, causing the set over-approximation to become vacuously large.

\begin{figure*}[t]
\centering
\includegraphics[width=0.95\textwidth, trim=50 85 50 40, clip]{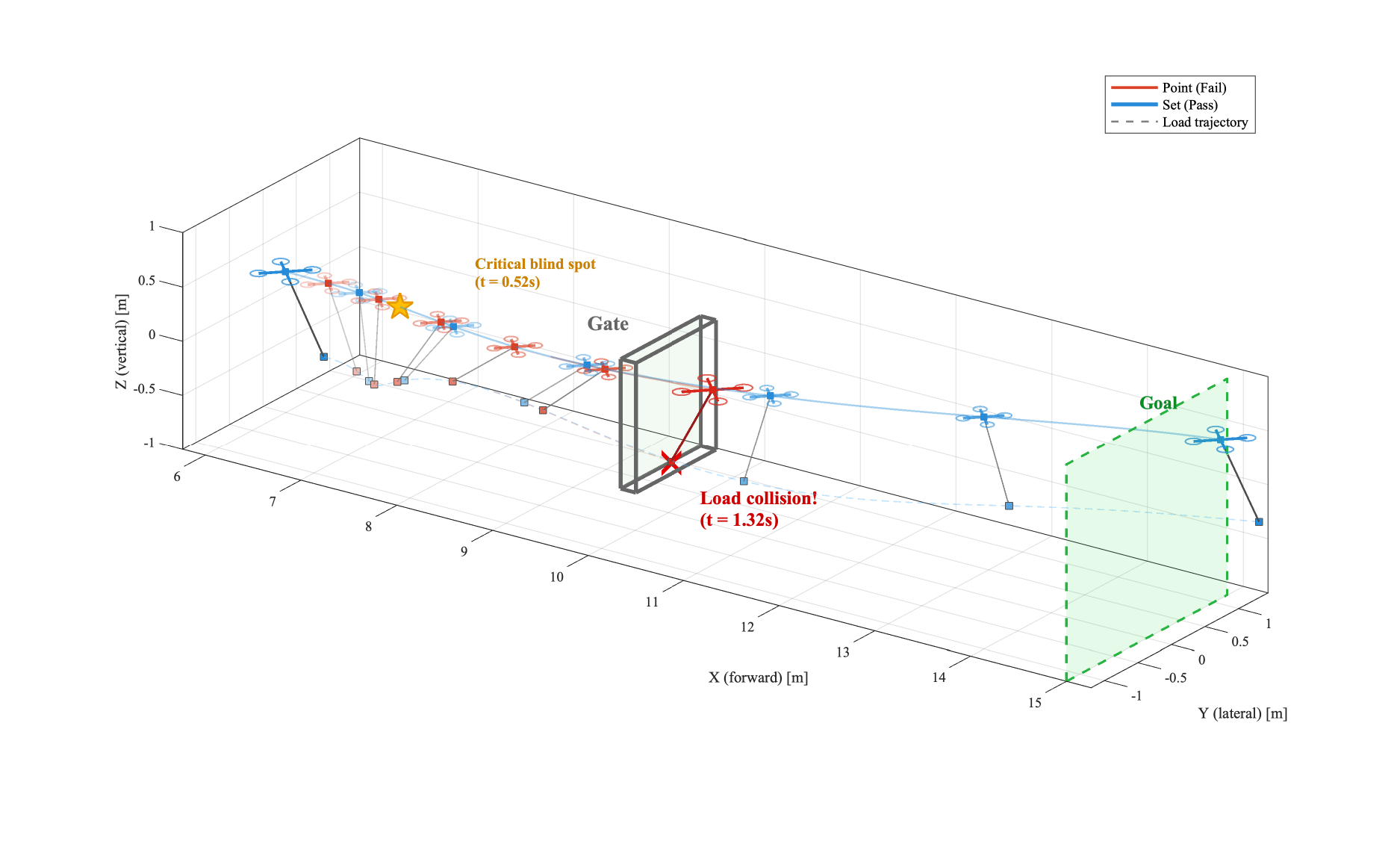}
\caption{3D gate passage comparison on the HARD scenario. SET (blue) passes safely; POINT (red) results in load collision at the gate.}
\label{fig:gate3d}
\end{figure*}

\section{Experiments}

\subsection{Setup}
The state uncertainty is modeled as a diagonal zonotope $\mathcal{X} = \langle x, \mathrm{diag}(\varepsilon) \rangle$ with $\varepsilon_i = 0.05$ for position and velocity states and $\varepsilon_i = 0.02$ for angular states and rates, representing conservative bounds on state estimation error at the 50\,Hz control rate.

The quadrotor must carry a suspended load (rod length 0.8 m) through a rectangular gate (1.2 m $\times$ 1.2 m opening) in 3D space. The task is to safely pass through the gate without collision between the gate frame and either the quadrotor body or the suspended load (Fig.~\ref{fig:gate3d}). The quadrotor approaches the gate along the positive $x$-axis with initial forward velocity $v_x = 2.0$--$2.5$ m/s. The full active control step, which includes encoder set-forward, certificate evaluation, physics-based $L_g$ computation, and QP solve, executes in 0.92 ms\footnote{All timing measurements on Apple M2, MATLAB R2024b.} (Table~\ref{tab:runtime}), accounting for 4.6\% of the 50 Hz control budget. Five canonical scenarios (Table~\ref{tab:scenarios}) test increasing difficulty:

\begin{table}[!htbp]
\caption{Computational cost breakdown at the gate.}
\label{tab:runtime}
\centering
\footnotesize
\begin{tabular}{lcc}
\toprule
\textbf{Component} & \textbf{Time (ms)} & \textbf{\% of 20 ms budget} \\
\midrule
Encoder set-forward (16D$\to$8D) & 0.28 & 1.40\% \\
Certificate eval + $L_g$ + QP solve & 0.64 & 3.20\% \\
\midrule
Full active control step & 0.92 & 4.60\% \\
\bottomrule
\end{tabular}
\end{table}

\begin{table}[!htbp]
\caption{Five canonical scenarios with increasing difficulty.}
\label{tab:scenarios}
\centering
\resizebox{\columnwidth}{!}{
\begin{tabular}{lccl}
\toprule
\textbf{Scenario} & $\Delta p_x$ & $\alpha_0$ / $\beta_0$ ($\dot{\alpha}_0$ / $\dot{\beta}_0$) & \textbf{Difficulty} \\
\midrule
GOOD & 7 m & 5\textdegree{} / 0\textdegree{} & Low \\
BAD & 3 m & 25\textdegree{} / 10\textdegree{} & Medium \\
HARD & 4 m & 30\textdegree{} / 0\textdegree{} & Med-High \\
HARD1 & 2 m & 35\textdegree{} / 0\textdegree{} ($\dot{\alpha}_0\!=\!40$\textdegree/s) & High \\
HARD2 & 2 m & 20\textdegree{} / 45\textdegree{} ($\dot{\beta}_0\!=\!30$\textdegree/s) & Extreme \\
\bottomrule
\end{tabular}}
\end{table}

\subsection{Set vs Point vs Point+Margin (Main Result)}
We compare three CBF evaluation modes on the same system with identical encoder, controller, and dynamics. The only difference is how the certificate $h$ is evaluated. Set-valued evaluation reports $h \geq 0$ in 44.4\% of per-head evaluations 
(1,477 out of 3,324 certificate evaluations across five scenarios), compared to 48.5\% for point evaluation (1,612 out of 3,324), confirming that set evaluation is more conservative. The 4.1\% gap (135 out of 3,324 evaluations) corresponds to blind spots where point evaluation incorrectly reports safety; this mechanism is analyzed in Section~\ref{sec:blind}. Table~\ref{tab:results} shows that this conservatism translates into superior actual safety:

\begin{itemize}
    \item SET (ours): $h_{\min} = w^\top c_z + b - \sum_j |w^\top G_{z,j}|$, adaptive per head and per step
    \item POINT: $h_{\min} = w^\top c_z + b$, evaluated at the center only
    \item PT\_MARGIN: $h_{\min} = w^\top c_z + b - \delta$ with fixed $\delta = 0.0310$, the global mean spread across all certificate heads and timesteps
\end{itemize}

\begin{table}[!htbp]
\caption{Gate passage results across five scenarios. $^\text{d}$divergence, $^\text{v}$vertical, $^\text{l}$lateral, $^\text{b}$body collision.}
\label{tab:results}
\centering
\resizebox{\columnwidth}{!}{
\begin{tabular}{lccccccc}
\toprule
\textbf{Mode} & \textbf{Score} & \textbf{GOOD} & \textbf{BAD} & \textbf{HARD} & \textbf{HARD1} & \textbf{HARD2} \\
\midrule
SET (ours) & 5/5 & $\checkmark$ & $\checkmark$ & $\checkmark$ & $\checkmark$ & $\checkmark$ \\
POINT & 1/5 & $\times^\text{d}$ & $\checkmark$ & $\times^\text{v}$ & $\times^\text{l}$ & $\times^\text{b}$ \\
PT\_MARGIN & 2/5 & $\checkmark$ & $\checkmark$ & $\times^\text{v}$ & $\times^\text{l}$ & $\times^\text{l}$ \\
\bottomrule
\end{tabular}}
\end{table}

Fig.~\ref{fig:gate3d} illustrates the HARD scenario ($\alpha_0 = 30^\circ$, 4 m from the gate). The set-based controller (blue) guides both the quadrotor and the suspended load safely through the gate, while the point-based controller (red) results in a load collision at the gate frame. The faint trails show the full trajectories; quadrotor and load icons are drawn at key timesteps with the collision moment frozen in red. The failure modes of POINT vary across scenarios: load-gate contact in the vertical direction (HARD), lateral direction (HARD1), quadrotor body collision (HARD2), and trajectory divergence (GOOD), confirming that the five scenarios exercise distinct safety-critical failure mechanisms that the set-based controller successfully avoids. The root cause of these failures is analyzed in Section V-C: point evaluation reports safety at moments when the state uncertainty set boundary is already unsafe. Section~V-C provides temporal evidence for this mechanism (Fig.~\ref{fig:htraces}): the last blind spot, where point evaluation still reports safety but set evaluation already detects danger, occurs at $t = 0.52$\,s, well before the gate plane. The set-based controller thus intervenes early enough to correct the trajectory, whereas the point-based controller detects the violation too late to correct and collides with the gate at $t = 1.32$\,s.

\subsection{Why Set Evaluation Helps: Blind Spot Analysis}
\label{sec:blind}
We analyze the gap between point and set evaluation along all trajectories:
\begin{equation}
\text{gap}(t) = h(E(x_c)) - \min_{z \in E(\mathcal{X})} h(z) = \sum_j |w^\top G_{z,j}|
\end{equation}

Over a total of 3324 certificate evaluations across all five scenarios, we find 135 instances (4.1\%) where point evaluation reports safety ($h_{\text{point}} \geq 0$) but set evaluation detects danger ($h_{\text{set}} < 0$). Table~\ref{tab:spread} details the per-head spread statistics. Fig.~\ref{fig:htraces} illustrates this on the HARD scenario: the point-based $h$ value (dashed) remains positive throughout, whereas the set-based $h_{\min}$ (solid) dips below zero near the gate, correctly detecting the unsafe region that ultimately leads to collision. This confirms the failure mechanism visible in Fig.~\ref{fig:gate3d}: without early detection of boundary violations, the point-based controller detects the violation too late to prevent the load from striking the gate frame. In contrast, the set-based controller detects these violations early and steers both the quadrotor and the load safely through the gate.

\begin{figure}[t]
\centering
\includegraphics[width=\columnwidth]{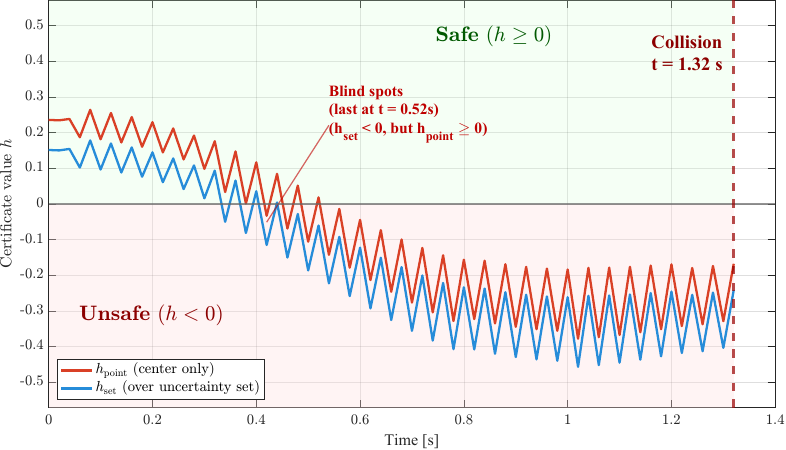}
\caption{Certificate traces $h_{\text{point}}$ (dashed) and $h_{\min} = \min_{z \in \mathcal{Z}} h(z)$ (solid) for the vertical clearance head $h_z$ on the HARD scenario. The shaded region marks the blind spot where point evaluation reports safe but set evaluation detects danger.}
\label{fig:htraces}
\end{figure}

\begin{figure*}[t]
\centering
\includegraphics[width=\textwidth]{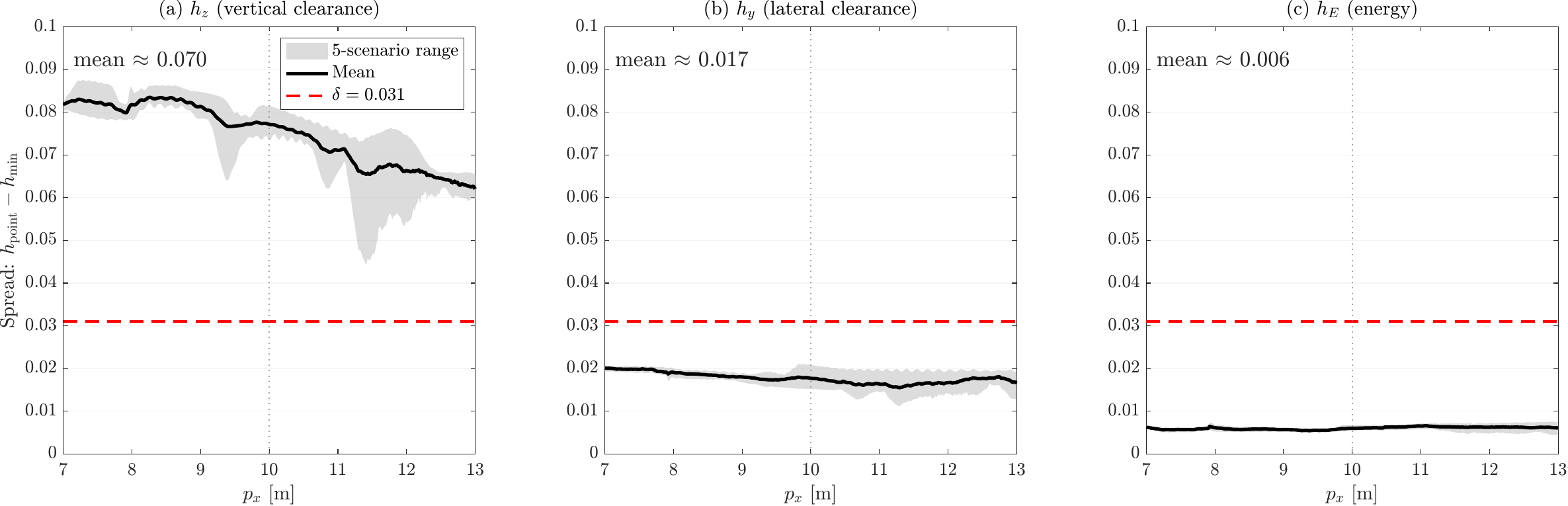}
\caption{Per-head certificate spread (gap between point and set evaluation) across all five scenarios as a function of position~$p_x$. Gray band: min--max range over the five scenarios; black line: mean. Red dashed line: fixed margin $\delta = 0.031$ used by PT\_MARGIN. The vertical dotted line marks the gate plane at $p_x = 10$\,m. The spread is highly anisotropic: $\delta$ falls below most of the $h_z$ spread, indicating insufficient coverage, while exceeding the $h_E$ spread by $5\times$, wasting control authority. Curves are interpolated onto a common $p_x$ grid and lightly smoothed for visualization; all reported metrics use unsmoothed data.}
\label{fig:spread}
\end{figure*}
\begin{table}[!htbp]
\caption{Per-head spread statistics between set and point evaluation.}
\label{tab:spread}
\centering
\footnotesize
\begin{tabular}{lccc}
\toprule
\textbf{Head} & \textbf{Mean spread} & \textbf{Max spread} & \textbf{Blind spots} \\
\midrule
$h_z$ & 0.0702 & 0.0943 & 119 \\
$h_y$ & 0.0171 & 0.0238 & 16 \\
$h_E$ & 0.0057 & 0.0089 & 0 \\
\bottomrule
\end{tabular}
\end{table}

The spread is highly anisotropic (Fig.~\ref{fig:spread}): $h_z$ exhibits $12.3\times$ larger mean spread than $h_E$ (0.070 vs.\ 0.006). Visually, wherever the spread curve lies above the red dashed line ($\delta = 0.031$), the fixed margin is insufficient to cover the gap between point and set evaluation, leaving blind spots undetected; conversely, wherever the curve lies well below the line, the margin over-compensates and wastes control authority. For $h_z$ (left panel), the spread exceeds $\delta$ throughout the approach, confirming that the fixed margin cannot guarantee set-level safety for this head. For $h_y$ and $h_E$ (center and right panels), $\delta$ exceeds the spread by $2\text{--}5\times$, indicating unnecessary conservatism. The wide gray band in the $h_z$ panel further reveals large variability across scenarios, meaning the gap fluctuates substantially with initial conditions. Moreover, the spread varies with position---peaking near the gate and decaying afterward---so even a per-head fixed margin would use the worst case across all timesteps, over-compensating at most moments. Set evaluation provides the minimal sufficient margin, adaptive per certificate head and per timestep. The dominance of $h_z$ spread indicates that point evaluation is most prone to underestimating safety risk along the vertical axis, consistent with the gravity-driven pendulum dynamics that amplify state uncertainty in this direction. The actual collision axis, however, depends on the scenario geometry: HARD fails vertically ($h_z$), HARD1 laterally ($h_y$), and HARD2 through body contact, confirming that set evaluation is necessary across all certificate heads.

\subsection{Conservatism: 16D Linear Certificate Limitations}
To understand why a nonlinear encoder is needed, we fit linear certificates $h_{\text{full}}(x) = w_{\text{full}}^\top x + b_{\text{full}}$ directly on the 16D state space using least-squares regression:

\begin{table}[t]
\caption{Linear certificate fit quality in the 16D state space.}
\label{tab:conservatism}
\centering
\footnotesize
\begin{tabular}{lc}
\toprule
\textbf{Head} & \textbf{16D linear $R$} \\
\midrule
$h_z$ (vertical) & 0.9962 \\
$h_y$ (lateral)  & 0.3841 \\
$h_E$ (energy)   & 0.7616 \\
\bottomrule
\end{tabular}
\end{table}

A linear certificate fit directly on the 16D state (Table~\ref{tab:conservatism}) achieves $R = 0.38$ for lateral clearance, indicating that a linear function in original space cannot capture the nonlinear dependence of gate clearance on coupled quadrotor position and pendulum angles (the load position involves $\sin\alpha\sin\beta$ terms). The learned encoder addresses this by extracting nonlinear features in which safety semantics become more accessible to linear certificate heads, and the set-based evaluation then provides worst-case safety certificates over the state uncertainty set within this learned representation.

\subsection{Limitations}
First, the full active control step adds 0.92 ms of computational overhead, which remains within the 20 ms budget at 50 Hz but may become prohibitive for faster control rates or deeper encoder architectures. Second, the directed dynamics error bound uses a 99.5\% quantile, which provides high-confidence but not formal worst-case coverage; a deterministic bound would require Lipschitz analysis of the latent dynamics model. Third, the numerical instantiation of Proposition 1 on the current model yields a conservative margin due to the conjugacy gap of the learned dynamics; tighter training or certified Lipschitz bounds could close this gap in future work. Finally, the safety guarantee of Proposition 1 assumes the system state remains within the domain where the encoder's approximate conjugacy holds; states far outside the training distribution may violate this assumption.

\section{Conclusion}
We extended the latent representation framework of~\cite{c3} to set-valued states, enabling worst-case safety certificates that cover entire state uncertainty sets rather than single points. The method combines sound zonotope propagation through the learned encoder with worst-case certificate evaluation over the resulting latent set. Experiments on a 16-dimensional quadrotor suspended-load system demonstrate that set-valued evaluation reports a lower (more conservative) safety rate than point evaluation (44.4\% vs.\ 48.5\% of per-head certificate evaluations reporting $h \geq 0$). This 4.1\% gap corresponds to blind spots where point evaluation misses boundary violations. By detecting these violations earlier, set-based control intervenes in time to prevent collision and achieves 100\% collision-free passages, compared to 20\% for point-based and 40\% for a fixed-margin baseline. Moreover, the safety gap varies up to $12\times$ across certificate heads, confirming that set evaluation provides adaptive margins per head and per timestep that no single fixed threshold can replicate.

Future directions include deploying the set-valued framework on robotic hardware, scaling to higher-dimensional systems, and learning the dynamics model end-to-end within the set framework. Extending the approach to vision-based robotic systems, where perceptual state uncertainty is large and out-of-distribution safeguards are needed, is also of interest.

\bibliographystyle{IEEEtran}
\bibliography{refs}

\end{document}